%% file: MAIN.tex
\documentclass{CSML}

\pdfoutput=1

\usepackage{lastpage}

\def\dOi{13(4:20)2017}
\lmcsheading%
{\dOi}
{1--\pageref{LastPage}}
{}
{}
{Jun.~\phantom04, 2016}
{Nov.~29, 2017}
{}

\input{packages}

\input{style/thm}

\input{style/macros}

\begin{document}


\title[Multiplicative-Additive Proof Equivalence is \Logspace-complete]
  {Multiplicative-Additive Proof Equivalence is \Logspace-complete, \via Binary Decision Trees\rsuper*}

\author[M.~Bagnol]{Marc Bagnol}
\address{Department of Mathematics and Statistics, University of Ottawa}
\email{bagnol@phare.normalesup.org}

\titlecomment{{\lsuper*}Extended and updated journal version of the author's TLCA'15 article~\cite{bagnol15}, including
an \textit{erratum} to a wrong proof in it.}

\maketitle


\begin{abstract}
Given a logic presented in a sequent calculus, a natural question is that of equivalence of proofs: to determine whether two given proofs are
equated by any denotational semantics, \ie any categorical interpretation of the logic compatible with its cut-elimination procedure. 
This notion can usually be captured syntactically by a set of rule permutations.

Very generally, proofnets can be defined as combinatorial objects which provide canonical representatives of equivalence classes
of proofs. In particular, the existence of proof nets for a logic provides a solution to the equivalence problem of this logic.
In certain fragments of linear logic, it is possible to give a notion of proofnet with good computational properties,
making it a suitable representation of proofs for studying the \cut-elimination procedure, among other things.

It has recently been proved
that there cannot be such a notion of proofnets for the multiplicative (with units) fragment of linear logic, due to the
equivalence problem for this logic being \Pspace-complete.

We investigate the multiplicative-additive 
(without unit) fragment of linear logic and show it is closely related to binary decision trees: 
we build a representation of proofs based on binary decision trees, 
reducing proof equivalence to decision tree equivalence, and give a converse encoding of binary decision trees as proofs.
We get as our main result that the complexity of the proof equivalence problem of the  studied fragment
is \Logspace-complete.


\end{abstract}

\section{Introduction}
Writing a proof in a formal system, or a program following the grammar of a programming language,
one often finds different equivalent ways of describing the same object.

This phenomenon can be described at a semantic level: recall that a denotational semantics of a logic is a categorical interpretation
that is invariant \modulo its \cut-elimination procedure.
We want to say that two \cut-free proofs are equivalent if they are interpreted as the same morphism in any denotational semantics.

\begin{defi}
	We say that two \cut-free proofs in a logic $\problem L$ presented in sequent calculus are \emph{equivalent} if they are interpreted
	by the same morphism in any denotational semantics of $\problem L$.
\end{defi}

Consider now that the logic at hand is given as a sequent calculus. Then we can consider a more syntactic flavor of the same idea 
by looking at rewriting steps that yield equivalent proofs, \eg by be permuting two rules. For example, in the following 
proof\footnote{For now, we do not bother with the exchange rule for ease of presentation, but this will have to change when we get into more technical considerations.
See \autoref{rem_mistake}.}

\begin{prooftree*}
	\Theproof\pi{A,B\vdash D}
	\Theproof\mu{A,C\vdash D}
	\Infer{2}[$\dplus{}$]{A,B\oplus C\vdash D}
	\Infer{1}[$\linear$]{B\oplus C\vdash A\linear D}
\end{prooftree*}
the $\linear$ rule can be lifted above the $\dplus{}$ rule, yielding the equivalent
\begin{prooftree*}
	\Theproof\pi{A,B\vdash D}
	\Infer{1}[$\linear$]{B\vdash A\linear D}
	\Theproof\mu{A,C\vdash D}
	\Infer{1}[$\linear$]{C\vdash A\linear D}
	\Infer{2}[$\dplus{}$]{B\oplus C\vdash A\linear D}
\end{prooftree*}
In the Curry-Howard view on logic, proofs are seen as programs, the \cut-rule corresponds to composition and the 
\cut-elimination procedure to the evaluation mechanism. The existence of equivalent but formally different objects is an 
issue in that during \cut-elimination, we may have to switch between equivalent objects to be able to actually perform a
reduction step. Consider 
\begin{prooftree*}
	\Theproof\nu{\vdash B}
	\Infer1[$\lplus$]{\vdash B\oplus C}

	\Theproof\pi{A,B\vdash D}
	\Theproof\mu{A,C\vdash D}
	\Infer{2}[$\dplus{}$]{A,B\oplus C\vdash D}
	\Infer{1}[$\linear$]{B\oplus C\vdash A\linear D}
	
	\Infer{2}[\cut]{\vdash A\linear D}
\end{prooftree*}
In order to go on with \cut-elimination, one option\footnote{The standard solution would rather be to lift the
\cut above the $\linear$, but the end result would be the same.}
is to apply the permutation we just saw to the $\dplus{}$ and $\linear$ rules to get:
\begin{prooftree*}
	\Theproof\nu{\vdash B}
	\Infer1[$\lplus$]{\vdash B\oplus C}

	\Theproof\pi{A,B\vdash D}
	\Infer{1}[$\linear$]{B\vdash A\linear D}
	\Theproof\mu{A,C\vdash D}
	\Infer{1}[$\linear$]{C\vdash A\linear D}
	\Infer{2}[$\dplus{}$]{B\oplus C\vdash A\linear D}
	
	\Infer{2}[\cut]{\vdash A\linear D}
\end{prooftree*}
and be able to perform an actual $\oplus/\dplus{}$ elimination step.
Note how this situation mirrors the permutation we took as an example just above.
These steps where some reorganization of the proof but no actual elimination occurs
are called \emph{commutative conversions}. They make the study of the \cut-elimination procedure much more
intricate, as one has to work \modulo equivalence, and this equivalence cannot be oriented: 
at some point we need to have the \cut rule in scope and realize it commutes with itself, if we want the
procedure to be confluent (indeed \mallm \cut-elimination is only confluent \modulo permutations)
\begin{center}
\begin{prooftree}
	\Theproof{\nu}{\vdash B}
	\Theproof{\mu}{\vdash A}
	\Theproof{\pi}{A,B\vdash C}
	\Infer{2}[\cut]{B\vdash C}
	\Infer{2}[\cut]{\vdash C}
\end{prooftree}
~~$\malleq$~~
\begin{prooftree}
	\Theproof{\mu}{\vdash A}
	\Theproof{\nu}{\vdash B}
	\Theproof{\pi}{A,B\vdash C}
	\Infer{2}[\cut]{A\vdash C}
	\Infer{2}[\cut]{\vdash C}
\end{prooftree}
\end{center}
This is particularly problematic in view of fine-grained analysis of the \cut-elimination procedure,
\eg in terms of complexity.
In any case, we can define a decision problem following the discussion above, and later discuss its decidability,
complexity \etc

\begin{defi}[equivalence problem]
	Given a logic ${\problem L}$ presented in sequent calculus, we define the \emph{equivalence problem} of ${\problem L}$
	which we write $\problem{Leq}$, as the decision problem:
	\begin{center}
		\enquote{Given two ${\problem L}$ proofs $\pi$ and $\nu$, are they equivalent?}
	\end{center}
	We write $\pi\malleq\nu$ when two proofs are equivalent.
\end{defi}

\subsection{Proofnets and complexity of proof equivalence}

The theory of proofnets, introduced alongside linear logic \cite{girard87}, aims at tackling this issue by looking for
a canonical representation of proofs, usually graphical objects, with the idea that when one manipulates canonical
objects the commutative conversions automatically disappear.
More technically: we look for a data structure offering canonical
representatives of equivalence classes of proofs.
The absolute minimal requirement is a translation function $\transl(\cdot)$ from proofs to proofnets.
One usually requires a section $\sequ(\cdot)$, from proofnets to proofs (with $\transl\circ\sequ=\Id{}$) 
of $\transl(\cdot)$ which is usually called \emph{sequentialization},
and of course the ability to implement \cut-elimination natively on the proofnet side. All of these together
allows to use proofnets as proofs on their own accord, while having only $\transl(\cdot)$ tells us about
proof equivalence and nothing more.

In this article we will only look into the complexity of computing the $\transl(\cdot)$ function and will refer to it simply as 
the complexity of the proofnets at hand.
For the multiplicative without units fragment of linear logic \mllm, a satisfactory notion of proofnet 
can be defined: it indeed offers a representation of equivalence classes of proofs, while
the translation from proofs to proofnets can be computed in logarithmic 
space\footnote{Given a proof and looking at its conclusion sequent, one only needs to compute when two atoms are linked by an axiom link. Which can be 
done by going from the root of the proof to an eventual leaf (axiom rule), potentially stopping at a $\otimes$ rule
if the two atoms are sent to two different subproofs.}.

Contrastingly, the linear logic community has struggled to extend the notion of proofnets
to wider fragments: even the case of \mll (that is, \mllm plus the multiplicative units)
could not find a satisfactory answer. 
A recent result~\cite{hh14} sheds some light on this question.
Since proofnets, when they exist, are canonical representatives of equivalence classes of proofs, they offer
a way to solve the equivalence problem of the logic: simply translate the two proofs and check 
if the resulting proofnets are the same.
The authors show that $\problem{MLLeq}$ is actually a \Pspace-complete problem. Hence, there is no hope
for a satisfactory notion of low-complexity proofnet for this fragment


In this article, we consider the same question, but in the case of \mallm: the multiplicative-additive
without units fragment of linear logic. Indeed, this fragment has so far also resisted
the attempts to build a notion of proofnet that at the same time characterizes proof
equivalence and has basic operations of tractable complexity: we have either canonical nets of
exponential size \cite{hvg05} or tractable nets that are not canonical \cite{girard96}.
D.~Hughes and W.~Heijltjes \cite{hh16} recently argued that it is unlikely that one can devise a notion of
proofnets for \mallm that is at the same time canonical and \Ptime.

One might suspect that we have completeness for some untractable complexity
class\footnote{The problem was first discovered to be decidable \cite{cp04}. Later on the first notion of canonical proofnets
for \mallm~\cite{hvg05} implied it can be solved in exponential time and finally the \coNP bound~\cite{mt03} on \cut-elimination for the fragment,
which subsumes the equivalence problem, was the best known to date.},
as is the case for \mll. An obvious candidate in that respect would be \coNP: as we will see, one of the
two approaches to proofnets for \mallm is related to Boolean formulas, which equivalence
problem is widely known to be \coNP-complete. Moreover, the \cut-elimination problem (given two proofs
with \cut{}s, do they have the same normal form) of \mallm is already known to be \coNP-complete \cite{mt03}.

It actually turns out that this is not the case as we will show that the
equivalence problem of \mallm is \Logspace-complete. 
However the connection between the problem of proofnets for this fragment and the notion of binary decision tree 
established in the course of the proof provides a beginning of clarification on the matter.

\subsection{Binary Decision Trees}

The problem of the representation of Boolean functions is of central importance in circuit
design and has a large range of practical applications. Over the years, binary decision
diagrams \cite{wegener00} became the most widely used data structure to treat this question.
We will actually use here the simpler cousins of binary decision diagrams which do not
allow sharing of subtrees: binary decision trees (\bdt).

Let us set up a bit of notation and vocabulary concerning \bdt.

\begin{defi}\label{def_bdt}
	A binary decision tree (\bdt) is a a binary tree with leaves labeled by $\zero$ or $\one$, and internal vertices labeled by Boolean variables.

	We will only be manipulating \emph{free} \bdt: no variable appears twice on a path from the root to a leave.
\end{defi}

A \bdt represents a Boolean function in the obvious way: given a valuation of Boolean variables ($v:x_i\mapsto\zero/\one$),
go down from the root
following the left/right branch of each internal vertex according to the assignment. 
When a leaf is reached one has the output of the function. 

\begin{exa}\label{ex_xor}
	The following \bdt represents the $x \text{\,\tt XOR\,} y$ function:
\begin{center}
\begin{tikzpicture}
\node[circle,draw](x){$y$}
  child{node[circle,draw]{$x$} child{node{$\zero$}} child{node{$\one$}}} 
  child [missing]
  child{node[circle,draw]{$x$} child{node{$\one$}} child{node{$\zero$}}} 
  ;
\end{tikzpicture}
\end{center}
\end{exa}

\begin{nota}\label{not_bdt}
	We use the notation 
	$\itef x\phi\psi$
	to describe \bdt inductively, with the meaning: the root is labeled by the variable $x$, with left subtree $\phi$ and right subtree $\psi$.
	Moreover we write $\bar\phi$ for the \emph{negation} of $\phi$: take $\phi$ and swap the $\zero/\one$ at the leaves.
\end{nota}

Different trees can represent the same function, and we have an equivalence relation accounting for that.

\begin{defi}\label{def_bdteq}
	We say that two \bdt $\phi,\psi$ are \emph{equivalent} (notation $\phi\booleq\psi$) if they represent the same Boolean function.
\end{defi}

\begin{exa}\label{ex_or}
	The two following \bdt represent the boolean function $x\text{\,\tt OR\,} (\text{\,\tt NOT\,}y)$ and are therefore equivalent:
\begin{center}
\begin{tikzpicture}
\node[circle,draw](x){$x$}
  child{node[circle,draw]{$y$} child{node{$\one$}} child{node{$\zero$}}} child{node{$\one$}}
  ;
\end{tikzpicture}
\qquad\qquad\qquad\qquad
\begin{tikzpicture}
\node[circle,draw](x){$y$}
  child{node{$\one$}}
  child{node[circle,draw]{$x$} child{node{$\zero$}} child{node{$\one$}}} 
  ;
\end{tikzpicture}
\end{center}
\end{exa}

As above, we can wonder about the complexity of the associated decision problem \bdteqprob. While the equivalence problem for Boolean formulas
is \coNP-complete in general and in many subcases including binary decision diagrams \cite
{wegener00}, 
we will see that the situation is different for \bdt due to their simpler structure.

\section*{Outline of the paper}

We begin in \autoref{sec_mall} by covering some background material on \mallm and notions of proofnet for this
fragment: monomial proofnets and the notion of slicing. Then, we introduce in \autoref{sec_bdtslice} an
intermediary notion of proof representation, \bdt slicings, that will help us to relate proofs in \mallm and
BDD. In \autoref{sec_reduction}, we establish further encodings and reductions.

\textit{The conference version \cite{bagnol15} of this work contained a mistake in a proof which we fix
in this version by considering a variant of \mallm where the exchange rule is taken seriously into account.
This is discussed further in \autoref{rem_mistake}}.


\section{Multiplicative-Additive Proof Equivalence} \label{sec_mall}
We will be interested in the multiplicative-additive without units part of linear logic, \mallm, and
more precisely its intuitionistic fragment. We choose to work primarily in this fragment for two reasons:
first we want to avoid any impression that our constructions rely on the classical nature of \mallm,
second we believe that this makes the whole discussion more accessible outside the linear logic community.
For reductions and encodings \emph{to} \mallm, this actually makes for stronger results
than if they were formulated in the classical case.
Conversely, when being limited the intuitionistic case would be restrictive (when deciding equivalence) 
we will explain how constructions can be extended to the classical case.

We consider formulas that are built inductively from atoms which we write $\alpha,\beta,\gamma,\dots$
and the binary connectives $\oplus$ and $\linear$.
We write formulas as $A,B,C,\dots$ unless we want to specify they are atoms and 
sequents as $\Gamma\vdash A$ with $\Gamma$ a sequence of formulas.
The $\oplus$ occurring on the of the left of the $\vdash$ symbol may be labeled by boolean variables
($\labplus x$, we
omit the label when not relevant) 
and we assume that different occurrences of the connective in a sequent carry a different label.
We consider an $\eta$-expanded logic,
which simplifies proofs and definitions. We do not include the \cut rule in our study, 
since in a static situation (we are not looking at the \cut-elimination procedure) it can always be encoded \wloss using
the $\dlinear$ rule. Moreover it can be argued that proof equivalence with \cut should include \cut-elimination,
which is known to be \coNP-complete \cite{mt03} as we already mentioned. 
Also we work with an explicit exchange rule, labeled by the positions $i,j$ of the elements of the context it swaps:

\smallskip
\begin{center}
	~
	\hfill
	\begin{prooftree}
		\Hypo{\alpha \vdash \alpha}
	\end{prooftree}
	\hfill
	\begin{prooftree}
		\Hypo{\Gamma,A\vdash B}
		\Infer{1}[$\linear$]{\Gamma \vdash A\linear B}
	\end{prooftree}
	\hfill
	\begin{prooftree}
		\Hypo{\Gamma\vdash A}
		\Hypo{B,\Delta\vdash C}
		\Infer{2}[$\dlinear$]{\Gamma,A\linear B,\Delta \vdash C}
	\end{prooftree}
	\hfill
	\begin{prooftree}
		\Hypo{\Gamma\vdash C}
		\Infer{1}[$\ex ij$]{\Gamma'\vdash C}
	\end{prooftree}
	\hfill
	~
	
	\medskip
	~
	\hfill
	\begin{prooftree}
		\Hypo{\Gamma\vdash A}
		\Infer{1}[$\lplus$]{\Gamma\vdash A\plus B}
	\end{prooftree}
	\hfill
	\begin{prooftree}
		\Hypo{\Gamma\vdash B}
		\Infer{1}[$\rplus$]{\Gamma\vdash A\plus B}
	\end{prooftree}
	\hfill
	\begin{prooftree}
		\Hypo{\Gamma, A\vdash C}
		\Hypo{\Gamma,B\vdash C}
		\Infer{2}[$\dplus{x}$]{\Gamma,A\labplus x B\vdash C}
	\end{prooftree}
	\hfill
	~
\end{center}

\smallskip


\begin{nota}
We will denote this proof system by \lamlp emphasizing that, even if we will
not explicitly manipulate \Lterms, we are working with the Curry-Howard counterpart of a typed linear 
\Lcalc with sums.

Moreover, we will use the notation \begin{prooftree}\Theproof\pi{\Gamma\vdash A}\end{prooftree}
for \enquote{the proof $\pi$ of conclusion $\Gamma\vdash A$}.
\end{nota}




\begin{rem}\label{rem_proofs}
	Any time we will look at a \lamlp proof from a complexity perspective, we will consider they are represented as
	a tree with nodes labelled with a rule and the sequent that is the conclusion of that rule.
\end{rem}

We already discussed the idea of proof equivalence in general, its synthetic formulation based on denotational semantics
and how it can be captured algorithmically as rule permutations in certain cases.
We will not go through all the details specific to the \mallm case, as we already have an available equivalent characterization
in terms of \emph{slicing} in the literature~\cite{hvg05,hvg15}
which we review in \autoref{sec_slicing}. Instead, in addition to the example we saw in the introduction
let us rather focus on another significant case:
\begin{center}
	\begin{prooftree}
		\Theproof\nu{\vdash D}
		\Theproof\pi{E,A\vdash C }
		\Theproof\mu{E,B\vdash C}
		\Infer{2}[$\dplus x$]{E,A\labplus x B\vdash C}
		\Infer{2}[$\dlinear$]{D\linear E, A\labplus x B \vdash C}
	\end{prooftree}
	\quad{\large$\malleq$}\qquad
	\begin{prooftree}
		\Theproof\nu{\vdash D}
		\Theproof\pi{E,A\vdash C}
		\Infer{2[$\dlinear$]}{D\linear E ,A\vdash C}
		\Theproof\nu{\vdash D}
		\Theproof\mu{E,B\vdash C}
		\Infer{2}[$\dlinear$]{D\linear E ,B\vdash C}
		\Infer{2}[$\dplus x$]{D\linear E ,A\labplus x B\vdash C}
	\end{prooftree}
\end{center}
In this permutation, the $\dlinear$ rule gets lifted above the $\dplus x$ rule. But doing so, notice
that we created two copies of $\nu$ instead of one, therefore the size of the prooftree has grown.
Iterating on this observation, it is not hard to build pairs of proofs that are equivalent, but
one of which is exponentially bigger than the other.
This \enquote{distributivity phenomenon} \cite{hh16} is indeed where the difficulty of proof equivalence in \lamlp (and hence \mallm) lies.
As a matter of fact, this permutation of rules alone would be enough to build the encoding of binary decision trees
by \lamlp proofs presented in \autoref{sec_reduction}.

	\subsection{Monomial proofnets} \label{sec_monomial}
The first attempt in the direction of a notion of proofnet for \mallm is due to J.-Y.Girard~\cite{girard96},
followed by a version with a full \cut-elimination procedure by O.~Laurent and R.~Maieli~\cite{lm08}.

While proofnets for multiplicative linear logic without units were introduced along with linear 
logic itself~\cite{girard87}, extending the notion to the multiplicative-additive without units fragment
proved to be a true challenge, mainly
because of the \emph{superposition} at work in the $\with$ rule ($\dplus{}$ in the intuitionistic fragment).

Girard's idea was to represent the superposed \enquote{versions} of the proof by a graph
with Boolean formulas (called a \emph{weight}) to each edge, with one variable for each $\with$ connective in the conclusion $\Gamma$.
To retrieve the version corresponding to some selection of the left/right branches
of each $\with$, one then just needs to evaluate the Boolean formulas with the corresponding valuation of their
variables, keeping only the parts of the graph which weight evaluates to $\one$. Here is an example of monomial
proofnet, from Laurent and Maieli's article:

\begin{center}
	\includegraphics[width=10cm]{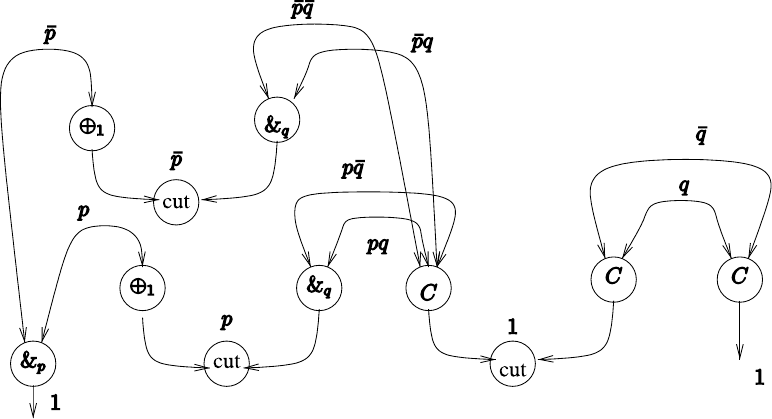}
\end{center}

This notion of proofnet fails to be a canonical representation of equivalence classes of proofs: for instance the rule 
permutation we just saw is not
interpreted as an equality. Still, we will retain this idea of using Boolean formulas and put it to work in
\autoref{sec_bdtslice}: suitably regrouping and merging links, it appears that instead of a collection of
scattered monomials we can use binary decision trees. This will lead us to the definition of
\bdt slicing.

	\subsection{Slicings and proof equivalence} \label{sec_slicing}
The idea of slicing dates back to J.-Y.~Girard's original article on proofnets for \mallm~\cite{girard96},
and was even present in the original article on linear logic~\cite{girard87}. It has been really developed,
addressing the many technical problems it poses by D.~Hughes and R.van Glabbeek~\cite{hvg05} to set up the only
known notion of canonical proofnets for \mallm.
It amounts to the natural point of view already evoked just above, seeing the $\with$ rule as introducing
superposed variants of the proof, which are eventually to be selected from in the course of \cut-elimination.

But if we have two alternative \emph{slices} for each $\with$ connective of a sequent $\Gamma$ and all
combinations of slices can be selected independently, we readily see that the global number of slices can be
exponential in the number of $\with$ connectives in $\Gamma$.
This is indeed the major drawback of the representation of proofs as a set of slices: the size of objects
representing proofs may grow exponentially with the size of the original proofs, impairing the 
possibility of proofnets based on this idea to be low-complexity. 

Let us give a variant of the definition
in the case of the intuitionistic \lamlp:

\begin{defi}[slicing]
	Given a \lamlp sequent $\Gamma\vdash A$, a \emph{slice} of $\Gamma\vdash A$ is a set
	of 
	(unordered) pairs of occurrences of atoms in $\Gamma\vdash A$.
	Then, a \emph{slicing} of $\Gamma\vdash A$ is a finite set of slices of $\Gamma\vdash A$, and
	to any \lamlp proof $\pi$, we associate a slicing $\slicing\pi$ by induction:

	\begin{itemize}
		\item If $\pi=\begin{prooftree}\Hypo{\alpha\vdash\alpha}\end{prooftree}$ then 
		      $\slicing\pi$ is the set containing only the linking $\big\{[\alpha,\alpha]\big\}$
		
		\item If $\pi=\,
		      \begin{prooftree}
		      \Theproof{\mu}{\Gamma,A\vdash B}
		      \Infer{1}[$\linear$]{\Gamma\vdash A\linear B}
		      \end{prooftree}$
		      \ \begin{minipage}[t]{10cm}
		      then $\slicing\pi=\slicing\mu$ \\
		      (seeing atoms 
		      of $A\linear B$ as the corresponding ones of $A$ and $B$)
		      \end{minipage}
		
		\bigskip
		\item If $\pi=\,
		      \begin{prooftree}
		      \Theproof{\mu}{\Gamma\vdash A}
		      \Theproof{\nu}{\Delta,B\vdash C}
		      \Infer{2}[$\dlinear$]{\Gamma,\Delta,A\linear B\vdash C}
		      \end{prooftree}$
		      then 
		      $\slicing\pi=\set{\lambda\cup\lambda'}{\lambda\in \slicing\mu\,,\,\lambda'\in\slicing\nu}$ 

		\bigskip
		\item If $\pi=\,
		      \begin{prooftree}
		      \Theproof{\mu}{\Gamma\vdash A}
		      \Infer{1}[$\ex ij$]{\Gamma'\vdash A}
		      \end{prooftree}$
		      then 
		      $\slicing\pi=\slicing\mu$ (accordingly reindexed)

		\bigskip
		\item If $\pi=\,
		      \begin{prooftree}
		      \Theproof{\mu}{\Gamma\vdash A}
		      \Infer{1}[$\lplus$]{\Gamma\vdash A\plus B}
		      \end{prooftree}$
		      then 
		      $\slicing\pi=\slicing\mu$,
		and likewise for $\rplus$
		
		\bigskip
		\item If $\pi=\,
		      \begin{prooftree}
		      \Theproof{\mu}{\Gamma,A\vdash C}
		      \Theproof{\nu}{\Gamma,B\vdash C}
		      \Infer{2}[$\dplus{}$]{\Gamma,\Delta,A\oplus B\vdash C}
		      \end{prooftree}$
		      then 
		      $\slicing\pi=\slicing\mu \cup \slicing\nu$
	\end{itemize}
\end{defi}

\begin{exa}
	The following proof $\pi$ (we index occurences of atoms to differentiate them)
	\begin{prooftree*}
		\Hypo{\alpha_l\vdash \alpha_r}
		\Infer{1}[$\lplus$]{\alpha_l\vdash\alpha_r\oplus\beta_r}
		\Hypo{\beta_l\vdash\beta_r}
		\Infer{1}[$\rplus$]{\beta_l\vdash\alpha_r\oplus\beta_r}
		\Infer{2}[$\dplus {}$]{\alpha_l\oplus\beta_l\vdash\alpha_r\oplus\beta_r}
	\end{prooftree*}
	gets the (two-slices) slicing $\slicing\pi=\big\{\{\,[\alpha_l,\alpha_r]\,\},\{\,[\beta_l,\beta_r]\,\}\big\}$.
\end{exa}

\begin{rem}
	In the $\dlinear$ rule, it is clear that the number of slices is multiplied. This is just what is
	needed in order to have a combinatorial explosion: pick any proof interpreted by two (or more) slices,
	and combine $n$ copies of this proof to get a proof that has linear size in $n$, but is interpreted with
	$2^n$ slices.
\end{rem}

The main result we need from the work of Hughes and van Glabbeek is their canonicity result~\cite{hvg05,hvg15}
formulated in the case of \mallm, that specializes readily to the intuitionistic fragment \lamlp.

\begin{thm}[slicing equivalence]\label{th_sliceeq}
	Let $\pi$ and $\nu$ be two \lamlp proofs.
	We have that $\pi$ and $\nu$ are equivalent if and only if $\slicing\pi=\slicing\nu$.
\end{thm}

\section{\BDT slicings} \label{sec_bdtslice}
We now introduce an intermediate notion of representation of proofs which will be a central tool
in the rest of the article.
In a sense, it is a synthesis of monomial proofnets and slicings: acknowledging the fact that slicing makes
the size of the representation explode, we rely on \bdt to keep things more compact.
Of course, the canonicity property is lost, but this is exactly the point! Indeed, deciding whether two
\enquote{\bdt slicings} are equivalent is the reformulation of proof equivalence we rely
on in the next sections.

The basic idea is, considering a sequent $\Gamma\vdash A$, to use the boolean variable associated to
$\oplus$ occurrences in $\Gamma$
and associate a \bdt to each pair of (occurrences of) atoms in $\Gamma\vdash A$, indicating
the presence of an axiom rule linking the two, depending on which branch of the $\dplus{}$ rules we are sitting in.

For example in the proof
\begin{prooftree*}
	\Hypo{\alpha\vdash \alpha}
	\Infer{1}[$\lplus$]{\alpha\vdash\alpha\oplus\beta}
	\Hypo{\beta\vdash\beta}
	\Infer{1}[$\rplus$]{\beta\vdash\alpha\oplus\beta}
	\Infer{2}[$\dplus x$]{\alpha\oplus_x\beta\vdash\alpha\oplus\beta}
\end{prooftree*}
we have a $\beta\vdash\beta$ axiom only when we are selecting the right branch of the $\dplus{}$ rule. So the \bdt $\itef x\zero\one$
 should be associated to the pair $[\beta,\beta]$ (here we use \autoref{not_bdt}),
and conversely $[\alpha,\alpha]$ gets $\itef x\one\zero$.

We can lift this idea to a complete inductive definition:

\begin{defi}[\bdt slicing]\label{def_boolslicing}
	Given a \lamlp sequent $\Gamma\vdash A$, a \emph{\bdt slicing} of $\Gamma\vdash A$ is a function $\slicef B$ 
	that associates a \bdt to every (unordered) pair 
	$[\alpha,\beta]$
	of distinct occurrences of atoms of $\Gamma\vdash A$.
	
	To any \lamlp proof $\pi$, we associate a \bdt slicing $\bslicing\pi$ by induction on the
	tree structure of $\pi$:
	\begin{itemize}
		\item If $\pi=\begin{prooftree}\Hypo{\alpha \vdash\alpha}\end{prooftree}$ then 
		      $\bslicing\pi[\alpha,\alpha]=\one$
		
		\bigskip
		\item If $\pi=\,
		      \begin{prooftree}
		      \Theproof{\mu}{\Gamma,A\vdash B}
		      \Infer{1}[$\linear$]{\Gamma \vdash A\linear B}
		      \end{prooftree}$
		      \ \begin{minipage}[t]{10cm}
		      then $\bslicing\pi[\alpha,\beta]=\bslicing\mu[\alpha,\beta]$ 
		      \\ (seeing atoms 
		      of $A\linear B$ as the corresponding atoms of $A$ and $B$)
		      \end{minipage}
		
		\bigskip
		\item If $\pi=\,
		      \begin{prooftree}
		      \Theproof{\mu}{\Gamma\vdash A}
		      \Theproof{\nu}{\Delta,B\vdash C}
		      \Infer{2}[$\dlinear$]{\Gamma,\Delta,A\linear B\vdash C}
		      \end{prooftree}$
		      then 
		      $\bslicing\pi[\alpha,\beta]=
		      \begin{cases}
		      \bslicing\mu[\alpha,\beta]\!\!\!\! &\text{ if $\alpha,\beta$ are atoms of $\Gamma\vdash A$}\\
		      \bslicing\nu[\alpha,\beta]\!\!\!\! &\text{ if $\alpha,\beta$ are atoms of $\Delta,B\vdash C$}\\
		      \zero &\text{ otherwise} 
		      \end{cases}
		      $
		
		\bigskip
		\item If $\pi=\,
		      \begin{prooftree}
		      \Theproof{\mu}{\Gamma\vdash A}
		      \Infer{1}[\tt ex]{\Gamma'\vdash A}
		      \end{prooftree}$
		      \ \begin{minipage}[t]{10cm}
		      then $\bslicing\pi[\alpha,\beta]=\bslicing\mu[\alpha,\beta]$ (accordingly reindexed)
		      \end{minipage}
		      
		\bigskip
		\item If $\pi=\,
		      \begin{prooftree}
		      \Theproof{\mu}{\Gamma\vdash A}
		      \Infer{1}[$\lplus$]{\Gamma\vdash A\plus B}
		      \end{prooftree}$
\begin{minipage}[t]{10cm}
		      then
		      $\bslicing\pi[\alpha,\beta]=
		      \begin{cases}
		      \bslicing\mu[\alpha,\beta] &\text{ if $\alpha,\beta$ are atoms of $\Gamma\vdash A$}\\
		      \zero &\text{ otherwise}
		      \end{cases}
		      $
		      \\
		      (and likewise for $\rplus$)
\end{minipage}
		
		\bigskip
		\item If $\pi=\,
		      \begin{prooftree}
		      \Theproof{\mu}{\Gamma,A\vdash C}
		      \Theproof{\nu}{\Gamma,B\vdash C}
		      \Infer{2}[$\dplus x$]{\Gamma,A\labplus x B\vdash C}
		      \end{prooftree}$\\\\
		      \hbox to 40 pt{\hfill}then
		      $\bslicing\pi[\alpha,\beta]= 
		      \begin{cases}
		      \itef x{\bslicing\mu[\alpha,\beta]}\zero &\text{ if $\alpha$ or $\beta$ is an atom of $A$}\\
		      \itef x\zero{\bslicing\nu[\alpha,\beta]} &\text{ if $\alpha$ or $\beta$ is an atom of $B$}\\
		      \itef x{\bslicing\mu[\alpha,\beta]}{\bslicing\nu[\alpha,\beta]} &\text{ if $\alpha$ or $\beta$ are both atoms of $\Gamma,C$}\\
		      \zero &\text{ otherwise} \\
		      \end{cases}
		      $
	\end{itemize}
\end{defi}

\noindent All cases except $\dplus{}$ are just ensuring proper branching of subformulas, while the
$\dplus{}$ case implements the idea behind the example before the definition: the left branch of the rule corresponds to the
left branch of the \bdt.

The equivalence of \bdt slicings is straightforward to define:

\begin{defi}
	We say that two \bdt slicings $\slicef M,\slicef N$ of the same $\Gamma\vdash A$ are \emph{equivalent}
	(notation $\slicef M\sliceeq \slicef N$) if for any pair $[\alpha,\beta]$
	we have $\slicef M[\alpha,\beta]\booleq \slicef N[\alpha,\beta]$ in the sense of \autoref{def_bdteq}.
\end{defi}

The usefulness of the notion of \bdt slicing comes from the fact that it captures exactly proof equivalence,
reducing it to \bdt equivalence.

\begin{thm}[proof equivalence]
	Let $\pi$ and $\nu$ be two \lamlp proofs, then $\pi\malleq\nu$ if and only if $\bslicing\pi\sliceeq\bslicing\nu$.
\end{thm}

\begin{proof}
	We will rely on the characterization of proof equivalence by slicing, \ie \autoref{th_sliceeq} and
	show that $\slicing\pi=\slicing\nu$ if and only if $\bslicing\pi\booleq\bslicing\nu$.
	
		\newcommand{\mB}{\mathcal B}
	
	To a \bdt slicing $\mB$, we can associate a slice $v(\mB)$ for each valuation $v$ of the variables
	occurring in $\mB$ by setting
	$v(\mB)=\set{[\alpha,\beta]}{v(\mathcal B[\alpha,\beta])=\one}$ (where $v(B[\alpha,\beta])$ denotes the
	result of evaluating a \bdt against a valuation of variables) and then a slicing
	$S_\mB=\set{v(\mB)}{v \text{ valuation}}$. Note that in the process different valuations yielding
	identical slices
	might be identified. By definition, it is clear that if $\mB$ and $\mB'$
	involve the same variables and $\mB\sliceeq\mB'$ then $S_\mB=S_{\mB'}$.
	
	It is straightforward to see that $S_{\bslicing\pi}=\slicing\pi$ so we have that $\bslicing\pi\booleq\bslicing\nu$
	implies $\slicing\pi=\slicing\nu$
	
	Conversely, suppose $\bslicing\pi\not\booleq\bslicing\nu$, so that there is a $v$ such that
	$v(\bslicing\pi)\neq v(\bslicing\nu)$. To conclude that $S_{\bslicing\pi}\neq S_{\bslicing\nu}$, we need
	the following lemmas:
	\begin{lemma}
		Let $\pi$ be a \lamlp proof of $\extlist{A_1}{A_n}\vdash A_0$ and $v$ a valuation of the variables of $\pi$ then 
		for any $i$, $v(\bslicing\pi)$ contains at least one pair with an atom in $A_i$.
	\end{lemma}
	\begin{lemma}
		Let $\pi$ be a \lamlp proof of $\Gamma\vdash C$ and $x$ a label of some $\dplus x$ rule introducing the
		subformula $A\labplus x B$ and $v$ a valuation of the variables of $\pi$ mapping $x$ to $\zero$ then:
		\begin{itemize}
			\item The pairs in $v(\bslicing\pi)$ contain no atom of $B$.
			\item At least one pair in $v(\bslicing\pi)$ contains an atom of $A$.
		\end{itemize}
\hspace{\parindent} (and conversely if $v$ maps $x$ to $\one$)
	\end{lemma}
\begin{proof}[Proof of the lemmas]
		The first one follows from a straightforward induction on $\pi$.

		The second is also by induction on $\pi$: going through the cases of \autoref{def_boolslicing} we see
		any rule that does not introduce $A\labplus x B$ itself but has it as a subformula will preserve these two properties.
		In the case of a $\dplus x$ rule branching two proofs $\mu$ and $\nu$, by definition any pair with atoms of $B$ will evaluate to
		$\zero$ if $x$ does, and hence cannot be part of $v(\bslicing\pi)$; conversely, the first lemma provides us
		with a pair containing an atom of $A$ in $v(\bslicing\mu)$ which will still be present in $v(\bslicing\pi)$.
\end{proof}

	\bigskip
	\noindent A consequence of the second lemma is that for any \lamlp proofs $\pi$ and $\nu$ with the same
	conclusion, and two \emph{different} valuations $v,w$ of their variables we have $v(\bslicing\pi)\neq w(\bslicing\nu)$:
	just apply the lemma with $x$ a variable on which $v,w$ differ.
	This means we have $v(\bslicing\pi)\neq v(\bslicing\nu)$ and also if $v\neq w$, $v(\bslicing\pi)\neq w(\bslicing\nu)$,
	so that $v(\bslicing\pi)\not\in S_{\bslicing\nu}$ and therefore $S_{\bslicing\pi}\neq S_{\bslicing\nu}$,
	that is $\slicing\pi\neq\slicing\nu$.
\end{proof}

This theorem will allow us to decide proof equivalence by reducing it to \bdt equivalence.
Therefore let us have a look at the complexity of computing \bdt slicings.

\begin{prop}\label{prop_logred}
	For any \lamlp proof $\pi$ and any pair $[\alpha,\beta]$,
	we can compute $\bslicing\pi[\alpha,\beta]$ in \Logspace.
\end{prop}

\begin{proof}
	To build $\bslicing\pi[\alpha,\beta]$ we only have to go through the prooftree and apply \autoref{def_boolslicing},
	creating new $x$ vertices when visiting $\dplus x$ rules, passing through $\linear$ and $\oplus$ rules,
	following only the relevant branch of $\dlinear$ rules \etc
	
	At any step we do not need to remember any other information than a pointer to our position in the prooftree,
	which requires only logarithmic space.
\end{proof}

\subsection{Proof equivalence is \Logspace-complete}\label{ssec_logspace}

Thanks to the results we just established, we will be able to decide \lamlp proof equivalence in logarithmic
space. Since \autoref{prop_logred} basically reduces proof equivalence to \bdt equivalence in logarithmic
space, let us focus on \bdt equivalence for a moment.

\begin{defi}[compatible leaves]
	Consider two leaves $l,m$ of two \bdt $\phi,\psi$ and the two paths $p,q$ from the leaves to the root of their
	respective trees. We say these leaves are \emph{compatible} if there is no variable that appears both in 
	$p$ and $q$ while being reached from opposite left/right branches.
\end{defi}

A small example to illustrate the idea: in the following trees $\phi$ (left) and $\psi$ (right)
\begin{center}
\begin{tikzpicture}
\node[circle,draw](x){$x$}
  child{node[circle,draw]{$t$} child{node[circle,draw]{$z$} child{node{$\zero$}} child{node{$\one$}}} child{node{$\zero$}}}
  child[missing]
  child{
    node[circle,draw]{$y$} child{node{$\zero$}} child{node{$\zero$}}
    };
\end{tikzpicture}
\qquad\qquad
\begin{tikzpicture}
\node[circle,draw](x)
{$t$}
	child{node[circle,draw]{$x$}
		child{node{$\one$}}
		child{node[circle,draw]{$z$} 
			child{node{$\one$}}
			child{node{$\zero$}}
		}
	}
	child[missing]
	child{node[circle,draw]{$y$} child{node{$\one$}} child{node{$\one$}}}
;
\end{tikzpicture}
\end{center}
the only $\one$ leaf of $\phi$ is compatible with the leftmost $\one$ leaf of $\psi$, but is not compatible with its
only $\zero$ leaf.

\begin{lemma}
	Two \bdt $\phi$ and $\psi$ are not equivalent if and only if there is a $\one$ leaf of $\phi$ and
	a $\zero$ leaf of $\psi$ that are compatible, or conversely.
\end{lemma}

\begin{proof}
	If we have compatible $\zero$ and $\one$ leaves, we can build a number of valuations of the variables
	that will lead to these leaves when evaluating $\phi$ and $\psi$, making them not equivalent.
	
	Conversely, if $\phi$ and $\psi$ are not equivalent, then we have a valuation that have them evaluate differently.
	Both of these evaluations are reached by following a path from the roots to leaves with different values, and these
	two leaves must be compatible.
\end{proof}

With this lemma we can devise a decision procedure in logarithmic space for \bdteqprob.

\begin{prop}\label{prop_bdteq}
	The \bdt equivalence problem is in \Logspace.
\end{prop}

\begin{proof}
	With the above lemma, we see that given two \bdt $\phi,\psi$ all we have to do is look for compatible leaves
	with different values. To do that we can go through all the pairs of one leaf of $\phi$ and one leaf of $\psi$
	(two pointers) and then through all pairs of variables on the paths between the leaves and the root of their
	respective trees (again, two pointers), checking if the two leaves hold different values 
	while being compatible. This needs four pointers of logarithmic size, so we can decide
	\bdteqprob in \Logspace.
\end{proof}

\begin{prop}
	The equivalence problem \lamlpeqprob is in \Logspace.
\end{prop}

\begin{proof}
	Combine \autoref{prop_logred} and \autoref{prop_bdteq}, relying on the fact that \Logspace algorithms can be composed without
	stepping out of \Logspace: given two \lamlp proofs $\pi,\mu$ with the same conclusion $\Gamma\vdash A$,
	go through all the pairs of atoms $[\alpha,\beta]$ of $\Gamma\vdash A$, compute the two associated \bdt
	$\bslicing\pi[\alpha,\beta],\bslicing\mu[\alpha,\beta]$ and test them for equivalence.
\end{proof}

To prove the hardness direction of the completeness result, we rely on a standard \Logspace-complete problem
and reduce it to \lamlpeqprob.

\newcommand{\ord}{\problem{ORD}\xspace}

\begin{defi}[\ord]\label{def_ord}
\emph{Order between vertices} (\ord) is the following decision problem:
	\begin{center}
	{\it\enquote{%
	Given a directed graph $G=(V,E)$ that is a line%
	\footnote{We use the standard definition of graph as a pair $(V,E)$ of sets of vertices and edges (oriented
	couples of vertices $x\rightarrow y$).
	A graph is a \emph{line} if it is connected and all the vertices have in-degree and out-degree $1$, except the
	\emph{begin} vertex which has in-degree $0$ and out-degree $1$ and the \emph{exit} vertex
	which has in-degree $1$ and out-degree $0$. A line induces a total order on vertices through its transitive
	closure}
	and two vertices $f,s\in V$\\
	do we have 
	$f<s$ in the total order induced by $G$?}}
	\end{center}
\end{defi}

\begin{thm} \ord is \Logspace-complete \cite{etessami97}.
\end{thm}

\begin{rem}
	Since we want to prove \Logspace-completeness results, the notion of reduction in what follows needs to be
	in a smaller complexity class than \Logspace itself (indeed \emph{any} problem in \Logspace is complete under \Logspace reductions).

One among many standard notions is (uniform) \aco reduction~\cite{csv84}, 
formally defined in terms of 
uniform circuits of fixed depth and unbounded fan-in. We will not be getting into the details
about this complexity class and, as we will consider only graph
transformations, 
we will rely on the following intuitive
principle: if a graph transformation locally replaces each vertex by a bounded number of vertices
and the replacement depends \emph{only} on the vertex considered and eventually its direct neighbors,
then the transformation is in \aco.
Typical examples of such a transformation are certain simple cases of so-called \enquote{gadget} reductions used in complexity theory
to prove hardness results.
\end{rem}\newpage

\begin{lemma}\label{lem_ord}
	\ord reduces to \lamlpeqprob in \aco.
\end{lemma}

\begin{proof}
	We are going to build a local graph transformation between the two problems.
	
	First, we assume \wloss that 
	the begin $b$
	vertex 
	of the line $G$ 
	is
	different from $f$ and $s$.
Consider the following proof (where the indexes are just used to identify occurences of the same atom $a$) 
that we call $\pi_{0}$:
\begin{center}
	\begin{prooftree}
		\Hypo{a_1\vdash a_1}
		\Hypo{a_2\vdash a_2}
		\Infer{2}[$\dlinear$]{a_1,a_1\linear a_2\vdash a_2}
		\Hypo{a_3\vdash a_3}
		\Infer{2}[$\dlinear$]{a_1,a_1\linear a_2,a_2\linear a_3\vdash a_3}
		\Hypo{a_4\vdash a_4}
		\Infer{2}[$\dlinear$]{a_1,a_1\linear a_2,a_2\linear a_3,a_3\linear a_4\vdash a_4}
	\end{prooftree}
\end{center}
Then encoding relies on the exchange rule: the begin vertex $b$ is replaced by the
proof $\pi_{0}$, the $f$ vertex is replaced by the $\ex 23$ rule with conclusion $a,a\linear a,a\linear a,a\linear a\vdash a$,
the $s$ vertex is replaced by the $\ex 34$ rule $a,a\linear a,a\linear a,a\linear a\vdash a$ and all the other
vertices are replaced by a sequence of two identical $\ex 12$ rules (so the end result is doing nothing) again with
conclusion $a,a\linear a,a\linear a,a\linear a\vdash a$. Edges are kept as they were.
Note that the conclusions of each \enquote{gadget} does
not depend on the rest of the tree (as is the case in general) which would be problematic with respect to \aco
complexity.

Because the permutations corresponding to $\ex 23$ and $\ex 34$ do not commute, asking if $f$ comes before $s$
amounts to asking if the resulting proof is equivalent to
\begin{center}
	\begin{prooftree}
		\Theproof{\pi_0}{a_1,a_1\linear a_2,a_2\linear a_3,a_3\linear a_4\vdash a_4}
		\Infer{1}[$\ex 23$]{a_1,a_2\linear a_3,a_1\linear a_2,a_3\linear a_4\vdash a_4}
		\Infer{1}[$\ex 34$]{a_1,a_2\linear a_3,a_3\linear a_4,a_1\linear a_2\vdash a_4}
	\end{prooftree}\vspace{-36 pt}
\end{center}
\end{proof}\vspace{12 pt}

\noindent So in the end, we get:

\begin{thm}[\Logspace-completeness]
	The problem \lamlpeqprob is
	\Logspace-complete.
\end{thm}

\begin{rem}
	Since the reduction we use for the hardness part uses only the rules $\dlinear$ and exchange, any
	subsystem of intuitionistic linear logic containing these rules will be have a \Logspace-hard equivalence problem.
	If they are moreover subsystems of \lamlp, the problem will be \Logspace-complete. For instance intuitionistic
	multiplicative linear logic (without units) has a \Logspace-complete equivalence problem.
\end{rem}

\begin{rem}\label{rem_mistake}
	The choice of including an explicit exchange rule needs to be commented here. In the the conference version of
	this article \cite{bagnol15} we tried to work \modulo exchange to simplify proofs and definitions.
	But when dealing with such low complexities as \Logspace and \aco reductions, this is not something we can
	afford: as we just saw, the exchange rule entails \Logspace-hardness of equivalence, so that working \modulo this rule
	amounts to work \modulo a problem which is hard for the complexity class we are looking at.
	
	This is at the root of a mistake in the conference version of this work: deprived of explicit exchange rule,
	we were forced to try to prove hardness \via a reduction of \bdt equivalence, which does not seem to
	be doable in \aco after all.
	More generally, this shows that when dealing with such low complexities the details of how proofs are implemented
	can be fully relevant.
\end{rem}

\subsection{Classical case} 
Exploring further the relation between \bdt and \lamlp proofs, let us
have a look at extending the above constructions and results to the classical case.

In \mallm we have one sided sequents $\vdash\Gamma$ (we even drop the $\vdash$ which no longer serve any purpose)
of formulas built from atoms and \emph{duals} of atoms $\nalpha, \nbeta, \ngamma, \dots$ and the connectives
$\oplus,\with,\otimes,\parr$. The $\with$ connectives are now holding the labels. Here are the rules of
\mallm:

\begin{center}
	~
	\hfill
	\begin{prooftree}
		\Hypo{\nalpha,\alpha}
	\end{prooftree}
	\hfill
\begin{prooftree}
		\Hypo{\Gamma,A,B}
		\Infer{1}[$\rparr$]{\Gamma,A\parr B}
	\end{prooftree}
	\hfill
	\begin{prooftree}
		\Hypo{\Gamma,A}
		\Hypo{\Delta,B}
		\Infer{2}[$\otimes$]{\Gamma,\Delta,A\otimes B}
	\end{prooftree}
	\begin{prooftree}
		\Hypo{\Gamma}
		\Infer{1}[$\ex ij$]{\Gamma'}
	\end{prooftree}
	\hfill
	\begin{prooftree}
		\Hypo{\Gamma,A}
		\Infer{1}[$\lplus$]{\Gamma,A\plus B}
	\end{prooftree}
	\hfill
	\begin{prooftree}
		\Hypo{\Gamma,B}
		\Infer{1}[$\rplus$]{\Gamma,A\plus B}
	\end{prooftree}
	\hfill
	\begin{prooftree}
		\Hypo{\Gamma,A}
		\Hypo{\Gamma,B}
		\Infer{2}[$\rlabwith{x}$]{\Gamma,A\labwith{x} B}
	\end{prooftree}
	\hfill
	~
\end{center}

Notice how close this is to what we had before. We are indeed looking at nothing more than a symmetrized version of \lamlp,
with a one-to-one correspondence between rules: $\parr$ matches $\linear$, $\otimes$ matches $\dlinear$ \etc
this makes the extension of previous result to \mallm extremely straightforward.

\Logspace-hardness is immediate since \lamlp is a subsystem of \mallm.

Moreover, the slicing and proof equivalence result of \autoref{sec_slicing} were originally formulated for \mallm~\cite{hvg05}
and we did nothing but adapt them to \lamlp. The notion of \bdt slicing never relies on the fact that we have
a separation $\Gamma\vdash A$ so extending this to \mallm is just a matter of giving the same interpretation to corresponding
rules. It follows that we can still reduce \mallmeqprob to \bdteqprob in logarithmic space, and therefore that
\mallmeqprob is in \Logspace.

\begin{thm}
	The equivalence problem \mallmeqprob is \Logspace-complete.
\end{thm}

\section{Reductions} \label{sec_reduction}
With the development of previous sections, we were able to show that \lamlpeqprob is \Logspace-complete.
This relied on the use of \bdt, which equivalence problem was show to be in \Logspace. In this section
we complete the picture by first showing that \bdt equivalence is \Logspace-hard and setting up a direct
reduction from \bdt to \lamlp proofs.

We start by showing that the \bdt equivalence problem is \Logspace-hard, by a similar argument to that of 
the proof of \autoref{lem_ord}.
Then we will show we have a converse transformation to \bdt slicing: we can encode any \bdt into
a \lamlp proof in logarithmic space, transporting \bdteqprob to \lamlpeqprob. This translation is more
expansive than the one we would get from completeness results, but it preserves the tree structure it manipulates.




	\subsection{\bdteqprob is \Logspace-complete}
\label{sec_redlin}
Let us now show that the equivalence problem of \bdt is \Logspace-hard (and hence complete),
again by reducing \ord (\autoref{def_ord}) to it. 




\begin{lemma}\label{lem_ord2}
	\ord reduces to \bdteqprob in \aco.
\end{lemma}

\begin{proof}
	(very similar to the proof of \autoref{lem_ord})
	First, we assume \wloss that the begin $b$ and the exit $e$ vertices of the line $G$ are different from $f$ and $s$. We write
	$f^+$ and $s^+$ the vertices immediately after $f$ and $s$ in $G$.
	
	Then, we perform a first transformation by replacing the graph with three copies of itself
	(this can be done by locally scanning the graph and create labeled copies of the vertices and edges).
	We write $x_i$
	to refer to the copy of the vertex $x$ in the graph $i$.
	The second transformation is a rewiring of the graph as follows: erase the edges going out of the $f_i$
	and $s_i$ and replace them as pictured in the two first subgraphs:
\begin{center}
\input{pic/pic_fs}
\qquad\qquad\qquad\qquad
\input{pic/pic_xy}

\end{center}
Let us call $G_r$ the rewired graph and $G_n$ the non-rewired graph. To each of them we add two binary vertices
$x$ and $y$ connected to the begin vertices $b_i$ as pictured in the third graph above.

Then we can produce two corresponding \bdt $\phi_r$ and $\phi_n$ by replacing the exit vertices $e_1$, $e_2$, $e_3$ by $\one$, $\zero$, $\zero$ 
respectively, and each complete each unary (those with an out-degree $1$) vertex with a $\rightarrow\!\one$ second
branch, making each vertex binary.

It is then easy to see that if $f<s$ in the order induced by $G$ if and only if $\phi_r$ and $\phi_n$ are equivalent.

Let us illustrate graphically what happens in the case where indeed $f<s$: we draw the resulting \bdt as a labeled
graph (with the convention that the we do not picture the extra $\rightarrow\!\one$ on unary vertices)
\begin{center}
\input{pic/pic_ord_ex}

\end{center}\vspace{-22 pt}
\end{proof}\medskip



\noindent As a consequence of the lemma and the results of \autoref{ssec_logspace} we get:

\begin{thm}
	The equivalence problems of
	\bdt (and hence of \bdt slicings) is
	\Logspace-complete.
\end{thm}




	\subsection{From \texorpdfstring{\BDT}{BDT} to \texorpdfstring{\lamlp}{lamdaopluslollipop} proofs (in \Logspace)}\label{sec_redlog}
We know now that both \bdteqprob and \lamlpeqprob are \Logspace-complete. Therefore there are reductions
in both directions between these two problems. However, coming from completeness theory
these reductions would not preserve the structure of the trees they manipulate.

The basic idea here is to use the $\dplus x$ rule to encode the $\itef x{\cdot}{\cdot}$ vertices.
We will have a formula $\boolf$ which will receive the encoded \bdt in the \bdt slicing interpretation of the proof,
and an atomic formula $\alpha_i$ for each variable.
To be able to mix the order in which the variables appear, we will need a careful treatment of which variables have
already been tested at any point.

\begin{defi}[free variable]
	Given a \bdt $\phi$ 
	we call the \emph{free variables} of a vertex or leaf $v$ of $\phi$ the list
	of variables one encounters on the path from $v$ to the root of $\phi$.
\end{defi}

In other words, a variable is free at a certain point of a \bdt if it needs to be tested to reach this point from the root of the tree.
For instance, in the \bdt
\begin{center}
	\begin{tikzpicture}
\node[circle,draw](x){$x$}
  child{node[circle,draw]{$t$} child{node[circle,draw]{$z$} child{node{$\zero$}} child{node{$\zero$}}} child{node{$\zero$}}}
  child[missing]
  child{
    node[circle,draw]{$y$} child{node{$\one$}} child{node{$\zero$}}
    };
\end{tikzpicture}
\end{center}
the $\one$ leaf has free variables $x,y$ while the $z$ vertex has free variables $x,t$.
Note that at any point of the tree, the two children of an internal vertex have the same free variables.

Let us then settle a few notation to help streamline the definition of the encoding.

\begin{nota}
	We fix atomic formulas $\extlist{\alpha_1}{\alpha_n} \dots$ and $\beta$ and write 
	$\boolf= \beta\plus\beta$. In what follows, we will use $\lbeta$ and $\rbeta$ to refer
	respectively to the left and right copies of $\beta$ in $\boolf$; 
	and likewise $\lalpha_i$ and $\ralpha_i$ for copies of $\alpha_i$ in $\alpha_i\oplus\alpha_i$
	and we implicitly associate the boolean variable $x_i$ to $\alpha_i\oplus\alpha_i$.
	Moreover, we write $\dalpha_i$ for the occurrence of $\alpha_i$ in $\form n=\alpha_n \linear \cdots \linear\alpha_2 \linear \alpha_1 \linear \beta$%
	\footnote{We follow the usual convention of writing the arrow $\linear$ as right-associating: 
	$\alpha \linear\beta \linear\gamma=\alpha \linear(\beta \linear\gamma)$}.
	
	Given $n$ and a subset $I\subseteq\extset{x_1}{x_n}$
	we can then identify uniquely atom occurrences in the sequents:
	\begin{itemize}
		\item $\contf I=\set{\alpha_i}{x_i\in I}$ 
		\item $\contb I=\set{\alpha_{j}\labplus{}\alpha_{j}}{x_j\not\in I}$ 
		\item $\cont nI=\contf I,\contb I,\form n$
	\end{itemize}
	
	We define respectively $\pright$ and $\pleft$ as the proofs
		\begin{prooftree}
			\Hypo{\beta\vdash\beta}
			\Infer{1}[$\lplus$]{\beta\vdash\boolf}
		\end{prooftree}
and
		\begin{prooftree}
			\Hypo{\beta\vdash\beta}
			\Infer{1}[$\rplus$]{\beta\vdash\boolf}
		\end{prooftree}.
	%
	%
	For any $k$, we write $\pplus k$ for the proof \ 
	\begin{prooftree}
		\Hypo{\alpha_k\vdash\alpha_k}
		\Hypo{\alpha_k\vdash\alpha_k}
		\Infer{2}[$\dplus{x_k}$]{\alpha_k\labplus{}\alpha_k\vdash\alpha_k}
	\end{prooftree}
	and $\pid k$ the proof (reduced to an axiom rule)
	\begin{prooftree}
		\Hypo{\alpha_k\vdash\alpha_k}
	\end{prooftree}.

Finally, we write $R,\exs$ when a rule is applied together with a series of exchanges before and after it
which are obvious from context.
%
\end{nota}

We can then go on with the definition of the encoding of a \bdt.

\begin{defi}\label{def_enc}
	Given a number of variables $n$, to any \bdt $\phi$ using the variables $\extlist{x_1}{x_n}$, 
	we associate a \lamlp proof $\repbdt{\phi}$ of conclusion $\cont n\void\vdash\boolf$ defined by
	induction on the tree structure of $\phi$. We think of any point $P$ in the tree as the root of a new \bdt, augmented with
	the information of the free variables at that point: to $P$ with free variables $I$ we associate a proof
	of conclusion $\cont nI\vdash\boolf$.
	\begin{itemize}
		\item If $P$ is a leaf $\zero/\one$ with free variables $I$, we begin by setting $I_k=I\cap\extlist{x_1}{x_k}$ and then define
		$\repbdt{P}$ as
	\begin{prooftree*}
		\Hypo{\nu_{n}}
		\Hypo{}
		\Hypo{\nu_2}
		\Hypo{\nu_1}
		\Theproof{\pi_{\zero/\one}}{\beta\vdash\boolf}
		\Infer{2}[$\dlinear$]{\cont 1{I_1}\vdash\boolf}
		\Infer{2}[$\dlinear$]{\cont 2{I_2}\vdash\boolf}
		\Infer[rule style=no rule]{2}{\reflectbox{$\ddots$}}
		\Infer[rule style=no rule]{1}{\cont {n-1}{I_{n-1}}\vdash\boolf}
		\Infer{2}[$\dlinear$]{\cont nI\vdash\boolf}
	\end{prooftree*}
	(with $\nu_i=\pid i$ if $x_i\in I$ and $\nu_i=\pplus i$ otherwise)
	
	\medskip
		\item If $P=\itef {x_k}QR$, with free variables $I$, then $P$ and $Q$ both have free variables $I'=I,x_k$ and we define $\pi_P$ as
		\begin{prooftree*}
		\Theproof{\pi_Q}{\cont n{I'}\vdash\boolf}
		\Theproof{\pi_R}{\cont n{I'}\vdash\boolf}
		\Infer{2}[$\dplus{x_k},\exs$]{\cont nI\vdash\boolf}
	\end{prooftree*}
	\end{itemize}
\end{defi}

\begin{exa}
	The \bdt for the $x\text{\,\tt OR\,} (\text{\,\tt NOT\,}y)$ function we saw in \autoref{ex_or} 
\begin{center}
\begin{tikzpicture}
\node[circle,draw](x){$x$}
  child{node[circle,draw]{$y$} child{node{$\one$}} child{node{$\zero$}}} child{node{$\one$}}
  ;
\end{tikzpicture}
\end{center}
	translates to the proof:

	\noindent
	\hspace{-2.6mm}
	\scalebox{0.8}{
	\begin{prooftree}
		\Hypo{\alpha_x\vdash\alpha_x}
		\Hypo{\alpha_y\vdash\alpha_y}
		\Theproof{\pi_{\one}}{\beta\vdash\boolf}
		\Infer{2}[$\dlinear$]{\alpha_y, \linear\alpha_y \linear \beta\vdash\boolf}
		\Infer{2}[$\dlinear$]{\alpha_x,\alpha_y, \alpha_x\linear\alpha_y \linear \beta\vdash\boolf}
		\Hypo{\alpha_x\vdash\alpha_x}
		\Hypo{\alpha_y\vdash\alpha_y}
		\Theproof{\pi_{\zero}}{\beta\vdash\boolf}
		\Infer{2}[$\dlinear$]{\alpha_y, \linear\alpha_y \linear \beta\vdash\boolf}
		\Infer{2}[$\dlinear$]{\alpha_x,\alpha_y, \alpha_x\linear\alpha_y \linear \beta\vdash\boolf}
		
		\Infer{2}[$\dplus{y},\exs$]{\alpha_x,\alpha_y\oplus\alpha_y, \alpha_x\linear\alpha_y \linear \beta\vdash\boolf}
		
		\Hypo{\alpha_x\vdash\alpha_x}
		\Theproof{\pplus y}{\alpha_y\labplus{}\alpha_y\vdash\alpha_y}
		\Theproof{\pi_{\one}}{\beta\vdash\boolf}
		\Infer{2}[$\dlinear$]{\alpha_y\oplus\alpha_y, \alpha_y \linear \beta\vdash\boolf}
		\Infer{2}[$\dlinear$]{\alpha_x,\alpha_y\oplus\alpha_y, \alpha_x\linear\alpha_y \linear \beta\vdash\boolf}

		\Infer{2}[$\dplus{x},\exs$]{\alpha_x\oplus\alpha_x,\alpha_y\oplus\alpha_y, \alpha_x\linear\alpha_y \linear \beta\vdash\boolf}
	\end{prooftree}
	}

\end{exa}
We then have to check that we get indeed a faithful encoding of \bdt{}s, reducing equivalence of
\bdt to equivalence of proofs.

\begin{lemma}[representation]
Writing $\slicef B$ the \bdt slicing of $\repbdt \phi$, we have

\begin{center}
	\begin{tabular}{ll}
	$\slicef B[\beta,\lbeta]=\phi$ &
	$\slicef B[\beta,\rbeta]={\bar{\phi}}$ \\[\parskip]
	$\slicef B[\lalpha_i,\dalpha_i]\booleq{\itef {x_i}\one\zero}$ &
	$\slicef B[\ralpha_i,\dalpha_i]\booleq{\itef {x_i}\zero\one}$
\end{tabular}
\end{center}
\end{lemma}

\begin{proof}
	A straightforward induction: at the leaf level (first part of the definition) we get, if we are dealing for instance with a $\one$ leaf,
	$\slicef B[\beta,\lbeta]=\one$, $\slicef B[\beta,\rbeta]=\zero$, if ${x_i}$ is not free then 
	$\slicef B[\alpha_i,\dalpha_i]=\one$, if on the contrary ${x_i}$ is free then $\slicef B[\lalpha_i,\dalpha_i]={\itef {{x_i}}\one\zero}$ and
	$\slicef B[\ralpha_i,\dalpha_i]={\itef {{x_i}}\zero\one}$.
	
	From there, we can see that the induction step for variable ${x_i}$ of the definition just branches together in the expected way the $\slicef B[\beta,\lbeta]$ and
	$\slicef B[\beta,\rbeta]$, does not change (up to equivalence) any of the $\slicef B[\,\cdot\,,\dalpha_j]$ for $j\neq i$ and turns
	$\slicef B[\alpha_i,\dalpha_i]\sim\one$ into $\slicef B[\lalpha_i,\dalpha_i]\sim{\itef {x_i}\one\zero}$ and
	$\slicef B[\ralpha_i,\dalpha_i]\sim{\itef {x_i}\zero\one}$.
\end{proof}

\begin{cor}
	Two \bdt $\phi,\psi$ are equivalent \iff $\repbdt \phi\malleq \repbdt \psi$.
\end{cor}

We can then have a look at the complexity of the reduction:

\begin{lemma}[\Logspace reduction]\label{lem_rep}
	The representation $\repbdt\phi$ of a \bdt $\phi$ by a proof can be computed in logarithmic space.
\end{lemma}

\begin{proof}
	We can do this in two steps: first process the \bdt to have the free variables at each vertex made explicit,
	which is done by listing the variables encountered between the vertex and the root of the tree and requires only
	to remember two positions in the tree; then each vertex of the new tree can be replaced by the corresponding piece of
	proof from \autoref{def_enc}. Both these steps can be performed in logarithmic space, and since logarithmic
	space function compose, we are done. 
\end{proof}

Some extra remarks about this encoding: first, because we need to keep track of which variables have been used
and which have not, we do not have a linear size bound but only a quadratic one. Second, on a more positive note:
we can see how closely the tree structure of $\pi_\phi$ mimics that of $\phi$; in fact we believe that
a form of one to one correspondance between proofs of $\cont nI\vdash\boolf$ and \bdt using variables
$\extlist{x_1}{x_n}$ could be worked out.


\section*{Conclusion} 
We established that the equivalence problem of intuitionistic additive-multiplicative (without units ) linear logic
is \Logspace-complete, this was achieved by introducing an intermediate representation of proofs based on binary
decision trees.

We also established low-complexity (computable in logarithmic space) 
correspondence between binary decision trees and proofs. This correspondance relates the tree structures
of proofs and \bdt very tightly amd ought to be studied further in view of potential limitation results.


However, the question of the possibility of a notion of proofnets for this logic is still unsettled:
even if the equivalence problem is in \Logspace, it might very well be that no notion of canonical representative
could be built.
The fact that \emph{optimization} of \bdt is a hard problem \cite{hr76} (even hard to approximate \cite{sieling08}),
could possibly be used to derive limitations on (if not impossibility of) the existence of low-complexity,
canonical representatives. But this has still to be clarified.

\section*{Acknowlegments} 

The author would like to thank the anonymous referees the their helpful and constructive comments,
and Willem Heijltjes for pointing to the idea of using the exchange rule to fix the \Logspace-hardness proof.

\bibliographystyle{plain}
\bibliography{biblio/biblio_AH,biblio/biblio_IZ} 

\end{document}

%% file: packages.tex
\usepackage[T1]{fontenc}		
\usepackage[utf8]{inputenc}		
\usepackage[english]{babel}				

\usepackage{tabularx}			
\usepackage{etex}				
\usepackage{etoolbox}
\usepackage{xargs,ifthen}	
\usepackage{fixltx2e} 			
\usepackage{fnpct}				

\usepackage{algorithmic}

\usepackage{tikz}
	\usetikzlibrary{calc,arrows,matrix,positioning,fit,decorations.markings}

\usepackage{multicol}		
\usepackage{wrapfig}		
\usepackage{graphicx}		
\usepackage{xcolor}			
\usepackage{aliascnt}

\usepackage{blindtext}
\usepackage{xspace}			
\usepackage{microtype}		
\usepackage{csquotes}		
\usepackage{relsize}		
\usepackage{ellipsis}

\usepackage{amssymb}			
\usepackage{mathtools}			
\usepackage{amsthm}
\usepackage{thmtools}

\usepackage{cmll}			
\usepackage{stmaryrd}		
\usepackage{ebproof}		
	\newcommand{\Theproof}[2]
		{\Hypo{\langle#1\rangle}\Infer[rule style=no rule]{1}{#2}} 
	\ebproofset{label separation=1pt} 
	\ebproofset{right label template={\footnotesize\inserttext}} 

\usepackage					
                [hidelinks]
	{hyperref}

\AtBeginDocument{
	
	\hyphenpenalty=1000
	\tolerance=400
	\setlength{\parskip}{2pt}
}

%% file: style/thm.tex
\newtheorem{theorem}{Theorem}[section]

\newtheorem{lemma}[theorem]{Lemma}

\newtheorem*{definition*}{Definition}



%% file: style/macros.tex
\newcommand{\bdt}{{\sf \textsf{BDT}}\xspace}
\newcommand{\BDT}{{\sf \textsf{BDT}}\xspace}

\newcommand{\zero}{\mathbf 0}
\newcommand{\one}{\mathbf 1}
\newcommand{\itef}[3]{#1\rhd #2 \talloblong #3}

\newcommand{\booleq}{\sim}

\newcommand{\bslicing}[1]{\slicef B_{#1}}
\newcommand{\slicing}[1]{\slicef S_{#1}}
\newcommand{\transl}{\mathbf t}
\newcommand{\sequ}{\mathbf s}
\newcommand{\lamlp}{$\lambda^\linear_\oplus$\xspace}
\newcommand{\malleq}{\sim}
\newcommand{\sliceeq}{\sim}
\newcommand{\mallm}{{\bf\textsf{MALL}}\textsuperscript{\!\bf-}\xspace}
\newcommand{\mllm}{{\bf\textsf{MLL}}\textsuperscript{\!\bf-}\xspace}
\newcommand{\mll}{{\bf\textsf{MLL}}\xspace}
\newcommand{\labwith}[1]{\with_{#1}}
\newcommand{\dplus}[1]{\oplus_{#1}^{\!\star}}
\newcommand{\labplus}[1]{\oplus_{#1}\!}
\newcommand{\dlinear}{\linear^{\!\star}}
\newcommand{\plus}{\oplus}
\newcommand{\lplus}{\oplus{\tt\scriptsize l}}
\newcommand{\rplus}{\oplus{\tt\scriptsize r}}
\newcommand{\rlabwith}[1]{\rwith_{#1}}

\newcommand{\mallmeqprob}{\problem{\mallm{}eq}\xspace}
\newcommand{\bdteqprob}{\problem{\bdt{}eq}\xspace}

\newcommand{\lamlpeqprob}{\problem{\lamlp{}eq}\xspace}
\newcommand{\repbdt}[1]{\pi_{#1}}
\newcommand{\pplus}[1]{\pi^{#1}_\oplus}
\newcommand{\pid}[1]{\pi^{#1}_{\sc Id}}
\newcommand{\pleft}{\pi_{\tt 0}}
\newcommand{\pright}{\pi_{\tt 1}}
\newcommand{\ex}[2]{\texttt{ex}_{#1,#2}}
\newcommand{\exs}{\texttt{ex}}

\newcommand{\cont}[2]{\Gamma_{#1,#2}}
\newcommand{\contb}[1]{\Delta_{#1}}
\newcommand{\contf}[1]{\Lambda_{#1}}
\newcommand{\form}[1]{F_{#1}}
\newcommand{\nalpha}{\lneg\alpha}
\newcommand{\nbeta}{\lneg\beta}

\newcommand{\ngamma}{\lneg\gamma}

\newcommand{\lbeta}{\beta^{\tt l}}
\newcommand{\rbeta}{\beta^{\tt r}}
\newcommand{\dalpha}{\alpha^{\linear}}
\newcommand{\lalpha}{\alpha^{\tt l}}
\newcommand{\ralpha}{\alpha^{\tt r}}

\newcommand{\modulo}{\emph{modulo}\xspace}
\newcommand{\cut}{{\tt cut}\xspace}

\newcommand{\Lcalc}{{\hbox{$\lambda\!$-cal}culus}\xspace}

\newcommand{\Lterms}{$\lambda\!$-terms\xspace}

\newcommand{\locit}[1]{\emph{#1}}
\newcommand{\eg}{\locit{e.g.}~}
\newcommand{\ie}{\locit{i.e.}~}
\newcommand{\via}{\locit{via}~}

\renewcommand{\iff}{\locit{iff}~}
\newcommand{\wloss}{\locit{w.l.o.g.}~}

\newcommand{\etc}{\locit{etc.}\xspace}

\newcommand{\compclass}[1]{{\sf\textbf{#1}}}		
\newcommand{\problem}[1]{{\normalfont\textsf{\textbf{#1}}}}		
\newcommand{\slicef}{\mathcal}		

\newcommand{\extlist}[2]{#1,\,\dots\,,#2}					
\newcommandx{\extset}[3][1=]{#1\{#2,\,\dots\,,#3#1\}}		
\newcommandx{\set}[3][1=]{#1\{\:#2 \ #1|\ #3\:#1\}}		
\newcommand{\Id}[1]{{\sf Id}_{#1}}										
\newcommand{\void}{\varnothing}									

\newcommand{\boolf}{\mathbf B}								

\newcommand{\lneg}[1]{{#1}^\star}				
\newcommand{\rwith}{\binampersand}
\newcommand{\rparr}{\bindnasrepma}
\renewcommand{\with}{\binampersand}
\renewcommand{\parr}{\bindnasrepma}

\newcommand{\linear}{\multimap}

\newcommand{\Logspace}{\compclass{Logspace}\xspace}						
\newcommand{\Ptime}{\compclass{Ptime}\xspace}							
\newcommand{\Pspace}{\compclass{Pspace}\xspace}

\newcommand{\coNP}{\compclass{coNP}\xspace}

\newcommand{\aco}{\compclass{AC\textsubscript 0}\xspace}

%% file: pic/pic_fs.tex
	\begin{tikzpicture}
\tikzstyle{localstyle}=[draw=black,thick,->]
\tikzset{every node/.style={inner sep=1pt,font=\small}}
	\matrix(m)[row sep=15pt,column sep=10pt]{%
  \node(f1){$f_1$}; & \node(f2){$f_2$}; & \node(f3){$f_3$}; 
&&&\node(s1){$s_1$}; & \node(s2){$s_2$}; & \node(s3){$s_3$}; \\
  \node(f1+){$f_1^+$}; & \node(f2+){$f_2^+$}; & \node(f3+){$f_3^+$};
&&&\node(s1+){$\,s_1^+$}; & \node(s2+){$s_2^+$}; & \node(s3+){$s_3^+$}; \\
};
\path [localstyle] (f1) to (f3+.110);
\path [localstyle] (f2) to (f1+);
\path [localstyle] (f3) to (f2+);

\path [localstyle] (s3) to (s2+.80);
\path [localstyle] (s2) to (s3+.110);
\path [localstyle] (s1) to (s1+);
	\end{tikzpicture}

%% file: pic/pic_xy.tex
	\begin{tikzpicture}
\tikzstyle{localstyle}=[draw=black,thick,->]
\tikzset{every node/.style={inner sep=0pt,font=\small,text depth=1pt}}
	\matrix(m)[row sep=8pt,column sep=6pt]{%
&&\node(x){$x$};\\
&&&\node[text height=5pt,text depth=2pt](y){$y$}; \\
\node(b1){$b_3$};&& \node(b2){$b_2$};&& \node[text height=7pt](b3){$b_1$}; \\};
\path [localstyle] (x.-50) to (y);
\path [localstyle] (x.-130) to (b1);
\path [localstyle] (y) to (b2);
\path [localstyle] (y) to (b3);
	\end{tikzpicture}

%% file: pic/pic_ord_ex.tex
\begin{tikzpicture}
	\tikzstyle{localstyle}=[draw=black,thick,->]
\tikzset{every node/.style={inner sep=0.5pt,font=\small,text depth=0pt,text height=4pt}}
	\matrix(m)[row sep=7pt,column sep=15pt]{%
&&\node(b1){$b$}; & \node(d1){$\cdots$}; & \node(f1){$f$};
&&\node(f1+){$f^+$}; & \node(dd1){$\cdots$}; & \node(s1){$s$};
&&\node(s1+){$s^+$}; & \node(ddd1){$\cdots$}; & \node(e1){$\one$}; \\
&\node(y1){$y$};\\
\node(x){$x$};&&\node(b2){$b$}; & \node(d2){$\cdots$}; & \node(f2){$f$};
&&\node(f2+){$f^+$}; & \node(dd2){$\cdots$}; & \node(s2){$s$};
&&\node(s2+){$s^+$}; & \node(ddd2){$\cdots$}; & \node(e2){$\zero$}; \\[1mm]
&\\
&&\node(b3){$b$}; & \node(d3){$\cdots$}; & \node(f3){$f$};
&&\node(f3+){$f^+$}; & \node(dd3){$\cdots$}; & \node(s3){$s$};
&&\node(s3+){$s^+$}; & \node(ddd3){$\cdots$}; & \node(e3){$\zero$}; \\
};
\path [localstyle] (b1) to (d1);
\path [localstyle] (d1) to (f1);
\path [localstyle] (b2) to (d2);
\path [localstyle] (d2) to (f2);
\path [localstyle] (b3) to (d3);
\path [localstyle] (d3) to (f3);
\path [localstyle] (f1) to (f2+.west);
\path [localstyle] (f2) to (f3+.west);
\path [localstyle] (f3.north east) to (f1+.south west);
\path [localstyle] (f1+) to (dd1);
\path [localstyle] (dd1) to (s1);
\path [localstyle] (f2+) to (dd2);
\path [localstyle] (dd2) to (s2);
\path [localstyle] (f3+) to (dd3);
\path [localstyle] (dd3) to (s3);
\path [localstyle] (s1) to (s2+.west);
\path [localstyle] (s2) to (s1+.west);
\path [localstyle] (s3) to (s3+.west);
\path [localstyle] (s1+) to (ddd1);
\path [localstyle] (ddd1) to (e1);
\path [localstyle] (s2+) to (ddd2);
\path [localstyle] (ddd2) to (e2);
\path [localstyle] (s3+) to (ddd3);
\path [localstyle] (ddd3) to (e3);
\path [localstyle] (x) to (y1);
\path [localstyle] (x) to (b3);
\path [localstyle] (y1) to (b1);
\path [localstyle] (y1) to (b2);

\end{tikzpicture}